\documentclass[onecolumn,dvipdfm,12pt]{IEEEtran}

\usepackage{amssymb}
\usepackage{amsfonts}
\usepackage{mathrsfs}
\usepackage{amsthm}
\usepackage{amsmath}
\usepackage{dsfont}
\usepackage[numbers, sort&compress]{natbib}
\usepackage{setspace}

 \newtheorem{thm}{Theorem}
 \newtheorem{cor}[thm]{Corollary}

 \newtheorem{prop}[thm]{Proposition}

\theoremstyle{definition}
\newtheorem{rem}[thm]{Remark}

\newtheorem{exm}[thm]{Example}
  \newcommand{\f}{\mathbb{F}_{q}}

\begin{document}

\title{Stopping Sets of Algebraic Geometry Codes\thanks{The first two authors are supported by the National Key Basic Research Program of China (973 Grant No. 2013CB834204), and the National Natural
Science Foundation of China (Nos. 61171082, 10990011 and 60872025).
This paper was presented in part at the 9th Annual Conference on Theory and Applications of Models of Computation,
Turing Centenary Meeting, Chinese Academy of Sciences, Beijing, China,
May 16-21, 2012.
}}
\author{Jun Zhang\thanks{Jun Zhang is with the Chern Institute of Mathematics, Nankai University, Tianjin, 300071, P.R. China.
e-mail: zhangjun04@mail.nankai.edu.cn},
Fang-Wei Fu\thanks{Fang-Wei Fu is with the Chern Institute of Mathematics and LPMC, Nankai University, Tianjin, 300071, P.R. China.
e-mail: fwfu@nankai.edu.cn}
and Daqing Wan\thanks{Daqing Wan is with the Department of Mathematics, University of California, Irvine, CA 92697-3875, USA.
e-mail: dwan@math.uci.edu}
}
\date{}
 \maketitle
 
 \setlength{\baselineskip}{20pt}

\begin{abstract}
Stopping sets and stopping set distribution of a linear code play an important role
in the performance analysis of iterative decoding for this linear code.
Let $C$ be an $[n,k]$ linear code over $\f$ with parity-check matrix $H$,
where the rows of $H$ may be dependent.
Let $[n]=\{1,2,\cdots,n\}$ denote the set of column indices of $H$.
A \emph{stopping set} $S$ of $C$
with parity-check matrix $H$ is a subset of $[n]$ such that the
restriction of $H$ to $S$ does not contain a row of weight $1$. The
\emph{stopping set distribution} $\{T_{i}(H)\}_{i=0}^{n}$ enumerates
the number of stopping sets with size $i$ of $C$ with parity-check
matrix $H$. Denote $H^{*}$ the parity-check matrix consisting of all the non-zero
codewords in the dual code $C^{\bot}$.
In this paper, we study stopping sets
and stopping set distributions of some residue algebraic geometry (AG)
codes with parity-check matrix $H^*$. First, we give two descriptions of stopping sets of residue
AG codes. For the simplest AG codes, i.e., the generalized Reed-Solomon
codes, it is easy to determine all the stopping sets. Then we
consider AG codes from elliptic curves. We use the group structure
of rational points of elliptic curves to present a complete
characterization of stopping sets. Then the stopping sets, the
stopping set distribution and the stopping distance of the AG code
from an elliptic curve are reduced to the search, counting and
decision versions of the subset sum problem in the group of rational
points of the elliptic curve, respectively. Finally, for some special cases, we determine the
stopping set distributions of AG codes from elliptic curves.
\end{abstract}
\begin{IEEEkeywords}
Stopping sets, stopping set distribution, stopping distance, algebraic geometry codes, elliptic curves, subset sum problem.
\end{IEEEkeywords}

\bigskip

\section{Introduction}

Let $C$ be an $[n,k,d]$ linear code over $\f$ with length $n$,
dimension $k$ and minimum distance $d$. Let $H$ be a
parity-check matrix of $C$, where the rows of $H$ may be
dependent. Let $[n]=\{1,2,\cdots,n\}$ denote the set of column indices of $H$.
A \emph{stopping set} $S$ of $C$
with parity-check matrix $H$ is a subset of $[n]$ such that the
restriction of $H$ to $S$, say $H(S)$, does not contain a row of weight $1$. The
\emph{stopping set distribution} $\{T_{i}(H)\}_{i=0}^{n}$ enumerates
the number of stopping sets with size $i$ of $C$ with parity-check
matrix $H$. Note that the empty set $\emptyset$ is defined as a stopping set and $T_0(H)=1$.
A number of researchers have
recently studied the stopping sets and stopping set distributions of linear codes,
e.g.,
see~\cite{AA,ETG,1603762,1715531,4036420,4039663,5208498,6006580,fu,xia,weber,
wadayama,orlitsky,laendner,LLMC,MM,krishnan,kashyap,henk,hehn,han,feldman,
esmaeili,di,ghaffar,tanner,Hollmann}. Stopping sets and stopping set distribution of a linear code
are used to determine the performance of this linear code under iterative
decoding~\cite{di}.

The \emph{stopping distance} $s(H)$ of $C$ with the parity-check
matrix $H$ is the minimum size of nonempty stopping sets. It plays
an important role in the performance analysis of the iterative
decoding, just as the role of the minimum Hamming distance $d$ of a
code for maximum-likehood or algebraic decoding.  Analogously to the
redundancy of a linear code, Schwartz and Vardy~\cite{1603762}
introduced the \emph{stopping redundancy} $\rho(C)$,  the minimal
number of rows in the parity-check matrix $H$ for the linear code $C$
such that the stopping distance $s(H)=d$, to characterize the
minimal ``complexity'' of the iterative decoding for the code $C$.
The stopping redundancy of some linear codes such as Reed-Muller codes, cyclic codes and maximal distance separable (MDS) codes have been studied recently~\cite{1603762,1715531,4039663,5208498,han}.

Note that the stopping distance, the stopping sets and stopping set distribution
depend on the choice of the parity-check matrix $H$ of $C$. Recall that $H^{*}$ is the parity-check matrix consisting of all non-zero
codewords in the dual code $C^{\bot}$. For any parity-check matrix
$H$, it is obvious that $T_{i}(H)\geqslant T_{i}(H^{*})$ for all
$i$, since $H$ is a
sub-matrix formed by some rows of $H^{*}$. Although the iterative decoding with the parity-check matrix
$H^{*}$ has the highest decoding complexity, it achieves the best
possible performance as it has the smallest stopping set distribution. It is known from \cite{weber} and \cite{henk}
that the iterative decoding with the parity-check matrix $H^{*}$ is an
optimal decoding for the binary erasure channel. The stopping set distribution is used to characterize the performance under iterative decoding. So it is important to determine the stopping set distribution of $C$ with the parity-check matrix $H^*$. However, in general, it is difficult to determine the stopping set distribution of $C$ with the parity-check matrix $H^*$. Using finite geometry, Jiang {\it et al.}~\cite{6006580} gave  characterizations of stopping sets of
some Reed-Muller codes (the Simplex codes, the Hamming cods, the first order Reed-Muller codes and the extended Hamming codes). And they determined the stopping set distributions of these codes. Since the iterative decoding with parity-check matrix
$H^{*}$ has the highest decoding complexity, they \cite{6006580} considered a parity-check matrix $H$, a submatrix of $H^*$, such that the stopping set distribution of $C$ with parity-check matrix $H$ is the same as that with $H^*$, but has the smallest number of rows. Such a parity-check matrix $H$ is called \emph{optimal} in certain sense. In general, it is difficult to obtain an optimal parity-check matrix for a general linear code. In \cite{6006580}, they obtained optimal parity-check matrices for the Simplex codes, the Hamming cods, the first order Reed-Muller codes and the extended Hamming codes. They also proposed an interesting problem to determine the stopping set distributions of well known linear codes with the parity-check matrix $H^*$. In this paper, we consider AG codes and a specific class of AG codes, i.e., AG codes associated with elliptic curves.

 \emph{From now on, we always choose the parity-check matrix $H^{*}$ for linear codes in this paper.} It is well-known that
\begin{prop}[\cite{1603762}]\label{stopdist} Let $C$ be a linear code with minimum distance $d(C)$,
and let $H^{*}$ denote the parity-check matrix for $C$ consisting of
all the nonzero codewords of the dual code $C^{\bot}$. Then the stopping distance
$s(H^{*})=d(C)$.
\end{prop}

Note that the generalized Reed-Solomon codes are MDS codes. For the $[n,k,d]$ MDS code $C$, i.e., $d=n-k+1$, its dual code $C^{\bot}$ is still an $[n,n-k,k+1]$ MDS code. Since any non-zero codeword in $C^{\bot}$ has at most $n-k-1$ zeros and any $(n-k)$ positions form an information set, we have

\begin{prop}\label{rs}
Let $C$ be an $[n,k,n-k+1]$ MDS code. Then
\begin{description}
\item[(i)]  any subset of $[n]$ with cardinality $\geqslant n-k+1$ is
a stopping set;

\item[(ii)] any non-empty subset of $[n]$ with cardinality $\leqslant n-k$
is not a stopping set.
\end{description}
\end{prop}
By Proposition \ref{rs}, we obtain the stopping set distribution of  MDS codes.
\begin{cor}
Let $C$ be an $[n,k,n-k+1]$ MDS code. Then
the stopping set distribution of $C$ is given by
\begin{displaymath}
T_{i}(H^{*}) = \left\{
\begin{array}{ll}
  1, & \mbox{if $i=0$,} \\
  0, & \mbox{if $1\leqslant i\leqslant n-k$,} \\
  \binom{n}{i}, & \mbox{if $i\geqslant n-k+1$\ .}
\end{array}\right.
\end{displaymath}
\end{cor}
As a generalization of the generalized Reed-Solomon codes, next we study
the stopping sets and stopping set distributions of AG codes.
\begin{flushleft}
\textbf{Constructions of AG codes.}
\end{flushleft}
 Without more special instructions, we fix some notation valid for
the entire paper.
\begin{itemize}
\item\emph{
 $X/\f$ is a geometrically irreducible smooth projective curve of genus $g$  over the finite field $\f$
with function field $\f(X)$. }
\item \emph{ $X(\f)$ is the set of all $\f$-rational points on $X$.}
\item\emph{ $D=\{P_{1},P_{2},\cdots,P_{n}\}$ is a proper subset of $X(\f)$.}
\item\emph{Without any confusion, also write
$D=P_{1}+P_{2}+\cdots+P_{n}$.}
\item\emph{ $G$ is a divisor of degree $m$ ($2g-2<m<n$) with
$\mathrm{Supp}(G)\cap D=\emptyset$.}
\end{itemize}
~~~Let $V$  be a divisor on $X$. Denote by $\mathscr{L}(V)$ the
$\f$-vector space of all rational functions $f\in \f(X)$ with the
principal divisor $\mathrm{div}(f)\geqslant -V$, together with the
zero function. And Denote by $\Omega(V)$ the $\f$-vector space of
all the Weil differentials $\omega$ with divisor
$\mathrm{div}(\omega)\geqslant V$, together with the zero
differential (cf.~\cite{Stichtenoth}). For any $\f$-rational point $P$ on $X$, choose one uniformizer
$t$ for $P$. Then for any differential $\omega$, we can write $\omega=u dt$ with some $u\in \f(X)$.
Write the $P$-adic expansion $u=\sum_{i=i_0}^{\infty}a_it^i$ for some $i_0\in \mathbb{Z}$ and $a_i\in \f$, the
residue map of $\omega$ at the point $P$ is defined to be
\[
  res_P(\omega)=res_{P,t}(u)=a_{-1}.
\]
 One can show that the above definition is well-defined~\cite[Proposition~4.2.9]{Stichtenoth}.

The residue AG code $C_{\Omega}(D, G)$ is defined to be the
image of the following residue map:
\begin{equation*}
    \begin{array}{cccl}
      res: & \Omega(G-D) & \rightarrow & \f^{n} \\
       & \omega & \mapsto & (res_{P_{1}}(\omega),res_{P_{2}}(\omega),\cdots,res_{P_{n}}(\omega))\ .
    \end{array}
\end{equation*}
 And its dual
code, the functional AG code $C_{\mathscr{L}}(D, G)$ is defined to
be the image of the following evaluation map:
\[
       ev: \mathscr{L}(G)\rightarrow \f^{n};\, f\mapsto
       (f(P_{1}),f(P_{2}),\cdots,f(P_{n}))\ .
\]
They are linear codes over $\f$, and have the code parameters $[n, n-m+g-1, d\geqslant m-2g+2]$ and
$[n, m-g+1, d\geqslant n-m]$, respectively. And they can be represented from each other~\cite[Proposition~8.1.2]{Stichtenoth}.

For the simplest AG codes, i.e., the generalized Reed-Solomon codes, we have determined all
the stopping sets. Then we consider the AG codes $C_{\Omega}(D, G)$ from elliptic curves. In this case, using the Riemann-Roch theorem, the stopping sets can be characterized completely as follows.

{\bf Main Theorem.}
Let $E$ be an elliptic curve over $\f$,
$D=\{P_{1},P_{2},\cdots,P_{n}\}$ a  subset of $E(\f)$ such that the
zero element $O\notin D$ and let $G=mO$ ($0<m<n$). The non-empty stopping sets
of the residue code $C_{\Omega}(D, G)$ are given as follows:
\begin{description}
\item[(i)]  Any non-empty subset of $[n]$ with cardinality $\leqslant m-1$ is
not a stopping set.

\item[(ii)] Any subset of $[n]$ with cardinality $\geqslant m+2$ is
a stopping set.

\item[(iii)] $A\subseteq [n]$, $\#A=m+1$, is a stopping set if and
only if for all $i\in A$, the sum
\[
\sum_{j\in A\setminus\{i\}}P_{j}\neq O\ .
\]

\item[(iv)] $A\subseteq [n]$, $\#A=m$, is a stopping set if and only
if
\[
      \sum_{j\in A}P_{j}=O\ .
\]
\item[(v)] Denote by $S(m)$ and $S(m+1)$ the two sets of stopping
sets with cardinality $m$ and $m+1$  in the cases (iv) and
(iii), respectively. Let
\[
  S^{+}(m)=\bigcup_{A\in S(m)}\{A\cup\{i\}: i\in [n]\setminus A\}\ .
\]
Then the union in $S^{+}(m)$ is a disjoint union, and we have
\[
   S(m+1)\cap S^{+}(m)=\emptyset\ ,
\]
and
\begin{equation*}
      S(m+1) =  \{\mbox{all subsets of $[n]$ with cardinality} m+1\}\setminus S^{+}(m)\ .
\end{equation*}
\end{description}
~~~The proof will be given in Section 3. By this theorem, the stopping set distribution of
$C_{\Omega}(D, G)$ follows immediately.
\begin{thm}\label{distofell}
Notation as above. The
stopping set distribution of $C_{\Omega}(D, G)$ with the
parity-check matrix $H^{*}$ is
\begin{displaymath}
T_{i}(H^{*}) = \left\{
\begin{array}{ll}
  1, & \mbox{if $i=0$,} \\
  0, & \mbox{if $1\leqslant i\leqslant m-1$,} \\
  \#S(m), &\mbox{if $i=m$,}\\
  \binom{n}{m+1}-(n-m)\#S(m), & \mbox{if $i=m+1$,} \\
  \binom{n}{i}, & \mbox{if $i\geqslant m+2$\ .}
\end{array}\right.
\end{displaymath}
\end{thm}

Then by Theorem~\ref{distofell}, we easily see that the stopping
distance of $C_{\Omega}(D, G)$ is $m$ or
$m+1$. But to decide it is equivalent to a decision
version of $m$-subset sum problem~\cite{liwan08,Klosters,liwan} in the group
$E(\f)$, which is an \textbf{NP}-hard problem under
{\bf RP}-reduction~\cite{chengqi}. Hence to compute the stopping distance
of $C_{\Omega}(D, G)$ is \textbf{NP}-hard under \textbf{RP}-reduction. To compute the
stopping set distribution is a counting version of
$m$-subset sum problem in the group $E(\f)$, so it is
also an \textbf{NP}-hard problem. But for a special $D\subseteq E(\f)$ with
strong algebraic structure, it is possible to compute the complete
stopping set distribution. For instance, if we take
$D=P\setminus \{O\}$, where $P$ is a subgroup of $E(\f)$. In particular, in application we always choose $D=E(\f)\setminus \{O\}$ to get a long linear code which is called standard elliptic code. Denote $N=\left|P\right|$ the cardinality of $P$, $\exp(P)$ the exponent of $P$, $P[d]$ the
$d$-torsion subgroup of $P$, and
\[
   N(m) =\frac{1}{N}\sum_{s|\exp(P)}(-1)^{m+\lfloor\frac{m}{s}\rfloor}\binom{N/s-1}{\lfloor m/s\rfloor}\sum_{d|s}\mu(s/d)\#P[d]\ ,
\]
respectively. It is known from ~\cite{liwan,Klosters} that $\#S(m)=N(m)$. Hence, we have
\begin{thm}\label{dist}
Let $D=P\setminus \{O\}$, where $P$ is a subgroup of $E(\f)$. The stopping set distribution of $C_{\Omega}(D, G)$ with the
parity-check matrix $H^{*}$ is
\begin{displaymath}
T_{i}(H^{*}) = \left\{
\begin{array}{ll}
  1, & \mbox{if $i=0$,} \\
  0, & \mbox{if $1\leqslant i\leqslant m-1$,} \\
  N(m), & \mbox{if $i=m$,}\\
  \binom{n}{m+1}-(n-m)N(m), &\mbox{if $i=m+1$,} \\
  \binom{n}{i}, &\mbox{if $i\geqslant m+2$\ .}
\end{array}\right.
\end{displaymath}
\end{thm}
This paper is organized as follows. In Section 2, we study stopping
sets of an arbitrary AG code and give algebraic and geometric descriptions of stopping
sets. In Section 3, we study the stopping sets and stopping
set distributions of AG codes $C_{\Omega}(D, G)$ from elliptic curves. We use the group structure
of rational points of elliptic curves to present a complete
characterization of stopping sets. It is shown that the stopping sets, the
stopping set distribution and the stopping distance of the AG code
from an elliptic curve can be reduced to the search, counting and
decision versions of the subset sum problem in the group of rational
points of the elliptic curve, respectively. We present the counting formula for
the stopping set distributions of AG codes from elliptic curves. In particular,
for some special cases, we determine explicitly the
stopping set distributions of AG codes from elliptic curves. Finally, some conclusions and open problems are given in Section 4.

\section{Stopping Sets of Algebraic Geometry Codes}
Let $X/\f$ be a geometrically irreducible smooth projective curve of genus $g$ over the finite field $\f$
with function field $\f(X)$, and $C_{\Omega}(D, G)$ the residue AG code from $X$. In this section, we study stopping sets and stopping set distributions of general residue AG codes and give algebraic and geometric descriptions of the stopping sets of $C_{\Omega}(D, G)$.
\begin{thm}\label{ss}
 A subset $A\subseteq [n]$ is a stopping
set of $C_{\Omega}(D, G)$ if and only if
\[
      \mathscr{L}(G-\sum_{j\in A}P_{j})=\bigcup_{i\in
      A}\mathscr{L}(G-\sum_{j\in A\setminus\{i\}}P_{j})\ .
\]
\end{thm}
\begin{proof}
By the definition, $A\subseteq [n]$ is  not a stopping set of $C_{\Omega}(D, G)$ if and
only if there is some $f\in \mathscr{L}(G)$ such that
\[
      ev(f)|_{A}=(f(P_{i}))_{i\in A}
\]
has weight $1$. That is, there is some $i\in A$ such that
\[
    f(P_{i})\neq 0  \qquad \mbox{and}\qquad f(P_{j})= 0 \,\,
    \mbox{for all}\, j\in A\setminus\{i\}\ .
\]
This is equivalent to saying that
\[
   f\in \mathscr{L}(G-\sum_{j\in A\setminus\{i\}}P_{j})\setminus \mathscr{L}(G-\sum_{j\in
   A}P_{j})\ .
\]
So $A$ is a stopping set if and only if for any $i\in A$,
\[
      \mathscr{L}(G-\sum_{j\in A}P_{j})=\mathscr{L}(G-\sum_{j\in A\setminus\{i\}}P_{j})\ .
\]
 Since $\mathscr{L}(G-\sum_{j\in A}P_{j})\subseteq
\mathscr{L}(G-\sum_{j\in A\setminus\{i\}}P_{j})$ for any $i\in A$,
we have
\begin{displaymath}
\begin{array}{cl}
   & \mathscr{L}(G-\sum_{j\in A}P_{j})=\bigcup_{i\in
      A}\mathscr{L}(G-\sum_{j\in A\setminus\{i\}}P_{j}) \\
 \Longleftrightarrow & \mathscr{L}(G-\sum_{j\in A}P_{j})=
\mathscr{L}(G-\sum_{j\in A\setminus\{i\}}P_{j})\ \mbox{for any $i\in A$.}
\end{array}
\end{displaymath}
 So the theorem holds.
\end{proof}
As a simple corollary, we obtain
\begin{cor}\label{cor}
 \emph{(i)} Any subset of $[n]$ with cardinality
$\geqslant m+2$ is a stopping set of $C_{\Omega}(D, G)$.

\emph{(ii)} Any non-empty subset of $[n]$ with cardinality $\leqslant m-2g+1$
is not a stopping set of $C_{\Omega}(D, G)$.
\end{cor}
\begin{proof} (i) For any subset  $A\subseteq [n]$ with cardinality $\geqslant
m+2$, divisors $G-\sum_{j\in A\setminus\{i\}}P_{j}$ and
$G-\sum_{j\in A}P_{j}$ are negative. So
\[
  \mathscr{L}(G-\sum_{j\in A}P_{j})=\bigcup_{i\in
      A}\mathscr{L}(G-\sum_{j\in A\setminus\{i\}}P_{j})=\{0\}\ .
\]
It follows from Theorem~\ref{ss} that $A$ is a stopping set.

(ii) For any non-empty subset $A\subseteq [n]$ with cardinality
$\leqslant m-2g+1$, by the Riemann-Roch theorem we have

\begin{displaymath}
\begin{array}{rcl}
  \dim(\mathscr{L}(G-\sum_{j\in A}P_{j})) & = & m-\#A-g+1, \\
  \dim(\mathscr{L}(G-\sum_{j\in A\setminus\{i\}}P_{j})) & = &
  m-\#A-g+2\ .
\end{array}
\end{displaymath}
So
\[
\mathscr{L}(G-\sum_{j\in A}P_{j})\varsubsetneq
\mathscr{L}(G-\sum_{j\in A\setminus\{i\}}P_{j})\
\]
for all $i\in A$.
It follows from Theorem~\ref{ss} that $A$  is not a stopping
set. Note that one can also give another proof of (ii) from Proposition~\ref{stopdist}, since the minimum distance of $C_{\Omega}(D, G)$
is at least $m-2g+2$.
\end{proof}

If we represent the generalized Reed-Solomon codes as AG codes from the rational function field, then by Corollary~\ref{cor}, we also obtain Proposition~\ref{rs} for the generalized Reed-Solomon codes.

Using the Riemann-Roch theorem, we give another description of
stopping sets of AG codes $C_{\Omega}(D, G)$.

\begin{thm}\label{equi}
A subset $A\subseteq [n]$ is a stopping set of $C_{\Omega}(D, G)$ if
and only if for any $i\in A$, there exists an effective divisor
$E_{i}$ with $P_{i}\notin \mathrm{Supp}(E_{i})$ such that
\[
    K-G+\sum_{j\in A}P_{j}\thicksim E_{i}\ ,
\]
where $K$ is a canonical divisor on $X$ and $\thicksim$ means that two
divisors are linearly equivalent, i.e., the difference between the two divisors is a principal divisor.
\end{thm}
\begin{proof}
From the proof of Theorem~\ref{ss}, a subset $A\subseteq [n]$ is a stopping set
if and only if for any $i\in A$,
$$\dim\mathscr{L}(G-\sum_{j\in
A}P_{j})=\dim\mathscr{L}(G-\sum_{j\in A\setminus\{i\}}P_{j})\ .$$

The Riemann-Roch theorem states that for any divisor $V$, we have
\[
    \dim\mathscr{L}(V)=\mathrm{deg}(V)-g+1+\dim\mathscr{L}(K-V)\ .
\]
So a subset $A\subseteq [n]$ is a stopping set if and only if for
any $i\in A$,
\[
\dim\mathscr{L}(K-G+\sum_{j\in
A}P_{j})=\dim\mathscr{L}(K-G+\sum_{j\in A\setminus\{i\}}P_{j})+1\ .
\]

It is equivalent to that for any $i\in A$, there exists
\[
     f\in \mathscr{L}(K-G+\sum_{j\in A}P_{j})\setminus\mathscr{L}(K-G+\sum_{j\in
A\setminus\{i\}}P_{j})\ .
\]
The last statement is equivalent to that for any $i\in A$, there
exists an effective divisor $E_{i}$ with $P_{i}\notin
\mathrm{Supp}(E_{i})
$ such that
\[
    K-G+\sum_{j\in A}P_{j}\thicksim E_{i}\ .
\]
Indeed, $E_{i}=\mathrm{div}(f)+K-G+\sum_{j\in A}P_{j}$.
\end{proof}
By Theorem~\ref{equi}, we immediately have a sufficient condition for a subset to be a stopping set.
\begin{cor}\label{suff}
Keep notation as above. Let $A$ be a subset of $[n]$. If $
K-G+\sum_{j\in A}P_{j}\thicksim E$ for some effective divisor $E$
whose support has no intersection with $\{P_{i}\,|\,i\in A\}$, then
$A$ is a stopping set.
\end{cor}

\section{Stopping Sets and Stopping set Distributions of AG Codes from Elliptic Curves}
In the previous section, for the general AG code $C_{\Omega}(D, G)$, we have seen that there is a gap, $\deg(G)-2g+2\leqslant i\leqslant \deg(G)+1$, where in general we
have not determined whether a subset with cardinality $i$ is a stopping set
or not. In this section, we consider a class of special AG codes, AG codes constructed from elliptic curves.

Let $X=E$ be an elliptic curve over the finite field $\f$ with a rational point $O$. Endow $E(\f)$ a group structure with the zero element $O$. Let
$D=\{P_{1},P_{2},\cdots,P_{n}\}$ be a subset of the set $E(\f)$ such that $O\notin D$. let
$G=mO$ ($0<m<n$).

In general, if $G$ is a divisor of degree $m$ on $E$, then for any rational point $Q\in E(\f)$, as $\deg(G-(m-1)Q)=1$, by the Riemann-Roch theorem, there exists
one and only one rational point $P\in E(\f)$ such that $G\sim (m-1)Q+P$. Suppose there exist rational points $Q,P$ such that
$G\sim (m-1)Q+P$ and $P,Q\notin D$. Let $G'=(m-1)Q+P$.
Then the codes $C_{\Omega}(D, G)$ and $C_{\Omega}(D, G')$ are
equivalent \cite[Proposition~2.2.14]{Stichtenoth}. And the dual codes $C_{\mathscr{L}}(D, G)$ and $C_{\mathscr{L}}(D,
G')$ are also equivalent. Here two linear codes $C_{1},C_{2}\subseteq \f^n$ are said to be
\emph{equivalent} if there is a vector $a=(a_{1},\cdots,a_{n})\in
(\f^*)^n$ such that
\[
  C_{2}=a\cdot C_{1}=\{(a_{1}c_{1},\cdots,a_{n}c_{n})\,|\,(c_{1},\cdots,c_{n})\in
  C_{1}\}\ .
\]
It is easy to see that two equivalent codes have the same stopping
sets and hence the same stopping set distributions. So to study the stopping sets and the stopping set distribution of
$C_{\Omega}(D, G)$, it suffices
to determine all the stopping sets and the stopping set distribution of
$C_{\Omega}(D, (m-1)Q+P)$. In this case, we use $Q$ to define the group $E(\f)$ with the zero element $Q$. Then all results in this paper hold
similarly for $C_{\Omega}(D, G)$ with $G\sim (m-1)Q+P$ such that $P,Q\notin D$.

Note that $g=1$ for elliptic curves. According to Corollary~\ref{cor}, any subset of $[n]$ with cardinality $\geqslant m+2$ is
a stopping set and any non-empty subset of $[n]$ with cardinality $\leqslant m-1$
is not a stopping set. So it is enough to consider the
subsets of $[n]$ with cardinality $m$ and $m+1$.
Below we use the group
$E(\f)$ \cite{DBLP:journals/jct/Schoof87,Silverman} to give a
description of these two classes of stopping sets with cardinality $m$ and $m+1$, respectively.

(i) Suppose $A\subseteq [n]$ with cardinality $m+1$ is not a
stopping set. Then there are some $i\in A$ and $f\in \mathscr{L}(G)$
such that
\[
 f\in \mathscr{L}(G-\sum_{j\in A\setminus\{i\}}P_{j})\setminus \mathscr{L}(G-\sum_{j\in
   A}P_{j})\ .
\]
Note that
\[
   \mathrm{deg}(G-\sum_{j\in A\setminus\{i\}}P_{j})=m-m=0\ ,
\]
and
\[
   \mathrm{div}(f)\geqslant -G+\sum_{j\in A\setminus\{i\}}P_{j}\ .
\]
Since both sides have degree zero,
so
\[
   \mathrm{div}(f)= -G+\sum_{j\in A\setminus\{i\}}P_{j}=\sum_{j\in A\setminus\{i\}}(P_{j}-O)\ .
\]
In this case,  $A\subseteq [n]$, $\#A=m+1$, is not a stopping set if
and only if there exists some $i\in A$ such that the sum $\sum_{j\in
A\setminus\{i\}}P_{j}$ in the group $E(\f)$ is $O$.

(ii) Suppose $A\subseteq [n]$ with cardinality $m$ is a
stopping set. By Theorem~\ref{ss}, for any $i\in A$, we have
\[
 \mathscr{L}(G-\sum_{j\in A\setminus\{i\}}P_{j})= \mathscr{L}(G-\sum_{j\in
   A}P_{j})\ .
\]
But
\[
   \mathrm{deg}(G-\sum_{j\in A\setminus\{i\}}P_{j})=1\geqslant 2g-1=1\ ,
\]
by the Riemann-Roch theorem, there exists some $f\in
\f(E)$ such that
\[
   0\neq f\in \mathscr{L}(G-\sum_{j\in A\setminus\{i\}}P_{j})= \mathscr{L}(G-\sum_{j\in
   A}P_{j})\ .
\]
So
\[
   \mathrm{div}(f)=G-\sum_{j\in A} P_{j}=\sum_{j\in A}(O- P_{j})\ .
\]
This is equivalent to
\[
   \sum_{j\in A}P_{j}=O
\]
in the group $E(\f)$.
Conversely, let $A\subseteq [n]$ with cardinality $m$ such that
$\sum_{j\in A}P_{j}=O$. Since the zero divisor $K=0$ is a canonical
divisor for elliptic curves, we have
\[
 K-G+\sum_{j\in A} P_{j}\sim 0\ .
\]
By Corollary~\ref{suff}, $A$ is a stopping set.

 From the argument above, we obtain the following partial results of the main theorem in the introduction.

\begin{thm}\label{main} Let $E$ be an elliptic curve over the finite
field $\f$, $D=\{P_{1},P_{2},\cdots,P_{n}\}$ a  subset of $E(\f)$
such that the zero element $O\notin D$ and let $G=mO$ ($0<m<n$). The non-empty
stopping sets of the residue code $C_{\Omega}(D, G)$ are given as
follows:
\begin{description}
\item[(i)]  Any subset of $[n]$ with cardinality $\leqslant m-1$ is
not a stopping set.

\item[(ii)] Any subset of $[n]$ with cardinality $\geqslant m+2$ is
a stopping set.

\item[(iii)] $A\subseteq [n]$, $\#A=m+1$, is a stopping set if and
only if for all $i\in A$, the sum
\[
\sum_{j\in A\setminus\{i\}}P_{j}\neq O\ .
\]

\item[(iv)] $A\subseteq [n]$, $\#A=m$, is a stopping set if and only
if
\[
      \sum_{j\in A}P_{j}=O\ .
\]
\end{description}
\end{thm}
Let us give an example to illustrate the theorem.

\begin{exm}\label{exm}
 Let $E$ be an elliptic curve defined over $\mathbb{F}_{5}$ by
the equation
\[
    y^2 =x^3 +x +1\ .
\]
Then $E$ has 9 rational points: the infinity point $O$ and
$P_{1}=(0, 1)$, $P_{2}=(4, 2)$, $P_{3}=(2, 1)$, $P_{4}=(3, 4)$,
$P_{5}=(3, 1)$, $P_{6}=(2, -1)$, $P_{7}=(4, -2)$, $P_{8}=(0, -1)$.
Using Group Law Algorithm 2.3 in \cite{Silverman}, one can
check that $E(\mathbb{F}_{5})$ forms a cyclic group with
$P_{i}=[i]P_{1}$. Let $D=\{P_{1},P_{2},\cdots,P_{8}\}$ and $G=3O$.

By Corollary~\ref{suff} and Theorem~\ref{main}, all nonempty stopping sets
of $C_{\Omega}(D, G)$ are given as follows:
\begin{description}
  \item[(i)] subsets of $[n]$ with cardinality $\geqslant 5$;
  \item[(ii)] \{1,2,3,7\}, \{1,2,3,8\}, \{1,2,4,5\}, \{1,2, 4,7\},
\{1,2,4,8\}, \{1,2,5,7\}, \{1,2,5,8\}, \{1,2,7,8\}, \{1,3,4,6\},
\{1,3,4,7\}, \{1,3,4,8\}, \{1,3,6,7\}, \{1,3,6,8\}, \{1,4,5,6\},
\{1,4,5,7\}, \{1,4,5,8\}, \{1,4,6,7\}, \{1,4,7,8\}, \{1,5,6,8\},
\{1,5,7,8\}, \{1,6,7,8\}, \{2,3,5,6\}, \{2,3,5,7\}, \{2,3,5,8\},
\{2,3,6,7\}, \{2,3,6,8\}, \{2,4,5,6\}, \{2,4,5,7\}, \{2,4,5,8\},
\{2,4,6,7\}, \{2,4,7,8\}, \{2,5,6,8\}, \{2,5,7,8\}, \{2,6,7,8\},
\{3,4,5,6\}, \{3,4,5,7\}, \{3,4,5,8\}, \{3,4,6,7\}, \{3,5,6,8\},
\{4,5,7,8\};
  \item[(iii)] \{1,2,6\}, \{1,3,5\}, \{2,3,4\}, \{3,7,8\}, \{4,6,8\},
\{5,6,7\}.
\end{description}

So the stopping set distribution of $C_{\Omega}(D, G)$ with the
parity-check matrix $H^{*}$ is

\begin{displaymath}
T_{i}(H^{*}) = \left\{
\begin{array}{cl}
  1, & \mbox{if $i=0$,} \\
  6, & \mbox{if $i=3$,} \\
  40, & \mbox{if $i=4$,} \\
  \binom{8}{i}, & \mbox{if $i\geqslant 5$,}\\
  0, & \mbox{otherwise}.
\end{array}\right.
\end{displaymath}
Also, the minimum distance of the code $C_{\Omega}(D, G)$
is 3 by Proposition~\ref{stopdist}.
\end{exm}
Theorem~\ref{main} describes all the stopping sets of residue AG
codes from elliptic curves. Next, we establish the relationship
between the set of stopping sets with cardinality $m$ and the set of
stopping sets with cardinality $m+1$.

Denote by $S(m)$ and $S(m+1)$ the two sets of stopping
sets with cardinality $m$ and $m+1$  in the cases (iv) and
(iii) in Theorem~\ref{main}, respectively.
Let $S^{+}(m)$ be
the extended set of $S(m)$ defined as follows
\[
  S^{+}(m)=\bigcup_{A\in S(m)}\{A\cup\{i\}: i\in [n]\setminus A\}\ .
\]
\begin{thm}\label{count}
 Notation as above. We have
\[
   S(m+1)\cap S^{+}(m)=\emptyset\ ,
\]
and
\[
     S(m+1)=\{\mbox{all subsets of $[n]$ with cardinality
     $m+1$}\}\setminus S^{+}(m)\ .
\]
Moreover, the union in the definition of $S^{+}(m)$ is a
disjoint union. Hence
\begin{equation*}
    \begin{array}{ccl}
       \#S(m+1) & = & \binom{n}{m+1}-\#S^{+}(m) \\
        & = & \binom{n}{m+1}-(n-m)\#S(m)\ .
     \end{array}
\end{equation*}
\end{thm}
\begin{proof}
First, $S(m+1)\cap S^{+}(m)=\emptyset$  is obvious by parts (iii) and (iv) of Theorem~\ref{main}. So
\[
     S(m+1)\subseteq \{\mbox{all subsets of $[n]$ with cardinality
     $m+1$}\}\setminus S^{+}(m).
\]
On the other hand, for any $A\notin S(m+1)$ and $\#A=m+1$, by
Theorem~\ref{main} (iii), there is some $i\in A$ such that $\sum_{j\in
A\setminus\{i\}}P_{j}= O$. By Theorem~\ref{main} (iv), $A\setminus\{i\}\in S(m)$. So
$$A=(A\setminus\{i\})\cup\{i\}\in
S^{+}(m)\ .
$$
 Hence
\[
     S(m+1)=\{\mbox{all subsets of $[n]$ with cardinality
     $m+1$}\}\setminus S^{+}(m)\ .
\]
If there exist $A\in S(m)$, $A^{\prime}\in S(m)$, $i\notin A$ and
$i^{\prime}\notin A^{\prime}$ such that
\[
  A\cup \{i\}=A^{\prime}\cup \{i^{\prime}\}\in S^{+}(m)\ .
\]
Then we have $i\in A^{\prime}$, $i^{\prime}\in A$ and $A\setminus
\{i^{\prime}\}=A^{\prime}\setminus \{i\}$.

Since
\[
  \sum_{j\in A}P_{j}=\sum_{j\in A^{\prime}}P_{j}=O
\]
we get $P_{i}=P_{i^{\prime}}$. So
\[
  A=A^{\prime},\qquad i=i^{\prime}\ .
\]
That is, the union in the definition of $S^{+}(m)$ is a
disjoint union. And the formula
\begin{equation*}
    \begin{array}{ccl}
       \#S(m+1) & = & \binom{n}{m+1}-\#S^{+}(m) \\
        & = & \binom{n}{m+1}-(n-m)\#S(m)\
     \end{array}
\end{equation*}
follows immediately.
\end{proof}
\begin{rem}
The above theorem shows how we can get $S(m+1)$ from $S(m)$.
Conversely, if we know $S(m+1)$, then by the above theorem, we
can exclude $S(m+1)$ from the set of all subsets of
$[n]$ with $m+1$ elements to get $S^{+}(m)$. For any $I\in S^{+}(m)$, we calculate
$\sum_{i\in I}P_{i}$. Then by Theorem~\ref{main} (iv), there is some index $j(I)\in I$ such that
\[
\sum_{i\in I}P_{i}=P_{j(I)}\ .
\]
By the definitions of $S(m)$ and $S^{+}(m)$, we have
 $$
 S(m)=\{I\setminus \{j(I)\}\,|\,I\in S^{+}(m)\}\ .
 $$
\end{rem}
In the above example, by Theorem~\ref{main} (iv), $S(3)$
consists of all the subsets of $[8]$ whose sums have $9$ as a
divisor. Then by Theorem~\ref{count}, $S(4)$ follows immediately from
$S(3)$.

The following corollary follows immediately from Proposition~\ref{stopdist}, Theorems~\ref{main} and~\ref{count}.
\begin{cor}
Notation as above. The minimum distance and the stopping distance of the residue AG code $C_{\Omega}(D, G)$ is
$\mathrm{deg}(G)$ or $\mathrm{deg}(G)+1$. More explicitly,
if $\#S(m)>0$, then we have the stopping distance
$$s(C_{\Omega}(D,G))=d(C_{\Omega}(D,G))=m=\mathrm{deg}(G)\ .
$$
If $\#S(m)=0$, then we have
$\#S(m+1)>0$ and hence
$$s(C_{\Omega}(D, G))=d(C_{\Omega}(D,G))=m+1=\mathrm{deg}(G)+1\ .
$$
\end{cor}

Let $\mathcal {G}$ be an abelian group with zero element $O$ and $D$
a finite subset of $\mathcal {G}$. For an integer $0<k<|D|$ and an
element $b\in D$ we denote
\[
  N_{\mathcal {G}}(k,b,D)=\#\{S\subseteq D\,|\,\#S=k\, \mbox{ and }\, \sum_{x\in S}x=b\}\ .
\]
Computing $N_{\mathcal {G}}(k,b,D)$ is called a counting version of
the \emph{$k$-subset sum problem} ($k$-SSP). In general, a counting
$k$-SSP is \textbf{NP}-hard~\cite{Cormen}. If there is no confusion, we simply
denote
\[
    N(k,b,D)=N_{\mathcal {G}}(k,b,D)\ .
\]

\begin{rem}
By the above theorem, for a general subset $D\subseteq
E(\f)$, to decide whether $\#S(m)>0$ is the decision $m$-subset sum
problem in $E(\f)$. It is known that the decision $m$-subset sum
problem in $E(\f)$ in general is {\bf NP}-hard under {\bf RP}-reduction~\cite{chengqi}. So to compute the stopping distance of $C_{\Omega}(D, G)$
is {\bf NP}-hard under {\bf RP}-reduction.
\end{rem}

But for a subset $D\subseteq E(\f)$ with special algebraic structure,
it is possible to give an explicit formula for $\#S(m)=N(m,O,D)$, and hence explicit formulae for
$\#S(m+1)$ and the whole stopping set distribution by
Theorem~\ref{count}. In the following, we consider special subsets $D=P\setminus\{O\}$ for some subgroup $P$ of $E(\f)$. In particular, recall that $C_{\Omega}(D, G)$ is called the standard elliptic code if $D=E(\f)\setminus \{O\}$.
\begin{prop}[\cite{liwan,Klosters}]\label{klosters}
Let $\mathcal {G}$ be a finite abelian group. For $b\in \mathcal {G}$, we have
\begin{equation*}
       N(i, b, \mathcal {G}\setminus\{0\})= \frac{1}{N}\sum_{s|\exp(\mathcal{G})}(-1)^{i+\lfloor\frac{i}{s}\rfloor}\binom{N/s-1}{\lfloor i/s\rfloor}
         \cdot\sum_{d|\gcd(e(b),s)}\mu(s/d)\#\mathcal {G}[d]\ .
\end{equation*}
where $N=\#\mathcal {G}$, $\exp(\mathcal {G})$ is the exponent of $\mathcal {G}$, $e(b) = \max\{d \,|\, d|\exp(\mathcal {G}),\, b \in d\mathcal {G}\}$, $\mu$ is the M\"obius function and $\mathcal {G}[d]$ is the $d$-torsion subgroup of $\mathcal {G}$.
\end{prop}

Set $\mathcal {G}=P$ a subgroup of $E(\f)$ in Proposition~\ref{klosters}. Let $N=|P|=n+1$ and $D=P\setminus\{O\}$. Then we have

\begin{thm}\label{distribution}
The number of stopping sets of $C_{\Omega}(D, mO)$ with cardinality $m$ is
\begin{equation*}
      \#S(m)= \frac{1}{N}\sum_{s|\exp(P)}(-1)^{m+\lfloor\frac{m}{s}\rfloor}\binom{N/s-1}{\lfloor m/s\rfloor} \\
         \cdot\sum_{d|s}\mu(s/d)\#P[d]\ .
\end{equation*}
\end{thm}

So together with Theorems~\ref{main} and~\ref{count}, we obtain Theorem~\ref{dist}.

It is well-known that the group $E(\f)$ of rational points is isomorphic to
\[
     E(\f)\cong \mathds{Z}/m_1\mathds{Z}\oplus \mathds{Z}/m_2\mathds{Z}\ ,
\]
for some integers $m_1|m_2$.
Then by
Theorems~\ref{main},~\ref{count} and \ref{distribution},
we can determine the stopping set distribution of the standard
residue AG code $C_{\Omega}(D, mO)$ from any elliptic curve $E/\f$
provided that we know the group structure of $E(\f)$. Explicitly, we
can compute $\#S(3)$ in Example~\ref{exm}:
\begin{displaymath}
\begin{array}{rl}
  \#S(3)= & \frac{1}{9}\sum_{s|9}(-1)^{3+\lfloor\frac{3}{s}\rfloor}\binom{9/s-1}{\lfloor 3/s\rfloor}\sum_{d|s}\mu(s/d)\left|\mathds{Z}/9\mathds{Z}[d]\right| \\
  =& \frac{1}{9}\left(\binom{8}{3}+\binom{2}{1}(3-1)-(9-3)\right)=6\ .
\end{array}
\end{displaymath}

So $\#S(4)=\binom{8}{4}-(8-3)\#S(3)=40$. This agrees with the
exhausting calculation in Example~\ref{exm}.

If we take special subgroups of $E(\f)$, then we have the following corollary.
\begin{cor}
 Notations as above.

 \emph{(i)} If we take
 \[
     P\cong \mathds{Z}/p^{t}\mathds{Z}
\]
 for some prime integer $p$ and integer $t\geqslant 1$, then
  \begin{equation*}
        \#S(m)=  \frac{1}{p^{t}}\left(\binom{p^{t}-1}{m}+(-1)^m(p^{t}-p^{\lfloor \log_{p}(m)\rfloor})\right.\\
        \left.+\sum_{i=1}^{\lfloor \log_{p}(m)\rfloor}(-1)^{m+\lfloor\frac{m}{p^i}\rfloor}(p^i-p^{i-1})\binom{p^{t-i}-1}{\lfloor\frac{m}{p^i}\rfloor}\right).
 \end{equation*}

 In particular, if $t=1$, then
 \[
     \#S(m)=\frac{1}{p}\left(\binom{p-1}{m}+(-1)^m(p-1))\right)\ .
 \]

 If $t=2$, then
 \begin{equation*}
        \#S(m)=   \frac{1}{p}\left(\binom{p-1}{m}+(-1)^m(p^2-p)+(-1)^{m+\lfloor\frac{m}{p}\rfloor}\right.\\
   \left.\cdot(p-1)\binom{p-1}{\lfloor\frac{m}{p}\rfloor}\right)\ .
 \end{equation*}

\emph{(ii)} If we take
 \[
     P\cong \mathds{Z}/p^{t_1}\mathds{Z}\oplus \mathds{Z}/p^{t_2}\mathds{Z}
\]
 for some prime integer $p$ and integers $1\leqslant t_1\leqslant t_2$, then
  \begin{equation*}
        \#S(m)=  \frac{1}{p^{t_1+t_2}}\left(\binom{p^{t_1+t_2}-1}{m}+\sum_{i=1}^{t_2}(-1)^{m+\lfloor\frac{m}{p^i}\rfloor}\right.\\
        \left.\cdot\binom{p^{t_1+t_2-i}-1}{\lfloor\frac{m}{p^i}\rfloor}(p^{i+\min\{i,t_1\}}-p^{i-1+\min\{i-1,t_1\}})\right)\ .
 \end{equation*}

\emph{(iii)} If we take
 \[
     P\cong \mathds{Z}/p_{1}^{t_1}\mathds{Z}\oplus \mathds{Z}/p_{2}^{t_2}\mathds{Z}
\]
 for two distinct prime integers $p_1,p_2$ and integers $t_1,t_2\geqslant 1$, then
 \begin{equation*}
    \begin{array}{cl}
        \#S(m)= & \frac{1}{p_{1}^{t_1}p_{2}^{t_2}}\left(\binom{p_{1}^{t_1}p_{2}^{t_2}-1}{m}+(p_{1}-1)(p_{2}-1)
        \cdot \sum\limits_{i=1}^{t_{1}}\sum\limits_{j=1}^{t_{2}}(-1)^{m+\lfloor\frac{m}{p_{1}^{i}p_{2}^{j}}\rfloor}p_{1}^{i-1}p_{2}^{j-1}\binom{p_{1}^{t_1-i}p_{2}^{t_2-j}-1}{\lfloor \frac{m}{p_{1}^{i}p_{2}^{j}}\rfloor}\right.\\
        &\left.+\sum\limits_{i=1}^{t_{1}}(-1)^{m+\lfloor\frac{m}{p_{1}^i}\rfloor}\binom{p_{1}^{t_1-i}p_{2}^{t_2}-1}{\lfloor\frac{m}{p_{1}^i}\rfloor}(p_{1}^{i}-p_{1}^{i-1})
        +\sum_{j=1}^{t_{2}}(-1)^{m+\lfloor\frac{m}{p_{2}^j}\rfloor}\binom{p_{1}^{t_1}p_{2}^{t_2-j}-1}{\lfloor\frac{m}{p_{2}^j}\rfloor}(p_{2}^{j}-p_{2}^{j-1})\right)\ .
     \end{array}
 \end{equation*}
 \end{cor}
\section{Conclusion}
In this paper, we study stopping sets
and stopping set distributions of residue algebraic geometry
codes $C_{\Omega}(D, G)$. Two descriptions of stopping sets of residue algebraic geometry
codes are presented. In particular,
there is a gap $\deg(G)-2g+2\leqslant i\leqslant \deg(G)+1$ where in general we
do not know whether a subset with cardinality $i$ is a stopping set
or not. In the case $g=0$, there is no gap and we have a complete
understanding. In the case $g=1$, using the group structure of
rational points of elliptic curves, we can characterize all the stopping
sets of algebraic geometry codes from elliptic curves. Then
determining the stopping sets, the stopping set distribution and the
stopping distance of $C_{\Omega}(D, G)$ are reduced to
$\mathrm{deg}(G)$-subset sum problems in finite abelian groups. In
the case $g>1$, only partial results can be obtained. It is still not known how to
compute the stopping set distribution.
For further work, there are two interesting problems:

(i) There are some papers contributing to compute the stopping redundancy of MDS codes \cite{1603762,4036420,5208498}. For AG codes from elliptic curves, we have seen that
the code is very closed to be MDS, i.e., MDS or near-MDS~\cite{shokrollahi} (an $[n,k,d]$ linear code is called \emph{near-MDS} if $d=n-k$ and the dual distance $d^{\bot}=k$). So how about the stopping redundancy of AG codes from elliptic curves?

(ii) In this paper, we have determined the stopping set distributions of AG codes from elliptic curves with the parity-check matrix $H^*$. Can we give optimal parity-check matrices for AG codes from elliptic curves?


\bibliographystyle{plain}
\bibliography{stoppingset}

\end{document}